\documentclass[conference,a4paper]{IEEEtran}

\usepackage{color}
\usepackage{amssymb}
\usepackage{amsthm}

\usepackage[utf8]{inputenc}
\usepackage[T1]{fontenc}
\usepackage{url}
\usepackage{ifthen}
\usepackage{cite}
\usepackage[cmex10]{amsmath} 

\usepackage{graphicx}
\newtheorem{theorem}{Theorem}

\newtheorem{corollary}{Corollary}

	\usepackage{cite}
	\usepackage[linesnumbered]{algorithm2e}

\newcommand{\Expect}{{\rm I\kern-.3em E}}

\begin{document}

\title{Coded Distributed Computing with Node Cooperation Substantially Increases Speedup Factors}

\author{%
  \IEEEauthorblockN{Emanuele Parrinello}
  \IEEEauthorblockA{EURECOM\\
                    Sophia Antipolis, France\\
                    Email: parrinel@eurecom.fr}
  \and
  \IEEEauthorblockN{Eleftherios Lampiris}
  \IEEEauthorblockA{EURECOM\\
                    Sophia Antipolis, France\\
                    Email: lampiris@eurecom.fr}
  \and
  \IEEEauthorblockN{Petros Elia}
  \IEEEauthorblockA{EURECOM\\
                    Sophia Antipolis, France\\
                    Email: elia@eurecom.fr}
}

\maketitle

\begin{abstract}
This work explores a distributed computing setting where $K$ nodes are assigned fractions (subtasks) of a computational task in order to perform the computation in parallel. In this setting, a well-known main bottleneck has been the inter-node communication cost required to parallelize the task, because unlike the computational cost which could keep decreasing as $K$ increases, the communication cost remains approximately constant, thus bounding the total speedup gains associated to having more computing nodes.  This bottleneck was substantially ameliorated by the recent introduction of coded MapReduce techniques which allowed each node --- at the computational cost of having to preprocess approximately $t$ times more subtasks --- to reduce its communication cost by approximately $t$ times. In reality though, the associated speed up gains were severely limited by the requirement that larger $t$ and $K$ necessitated that the original task be divided into an extremely large number of subtasks. In this work we show how node cooperation, along with a novel assignment of tasks, can help to dramatically ameliorate this limitation.  The result applies to wired as well as wireless distributed computing, and it is based on the idea of having groups of nodes compute identical parallelization (mapping) tasks and then employing a here-proposed novel D2D coded caching algorithm.

\end{abstract}

\section{Introduction}

Parallel computing exploits the presence of more than one available computing node, in order to allow for faster execution of a computational task. This effort usually involves dividing the original computational task into different subtasks, and then assigning these subtasks to different nodes which, after some intermediate steps, compute the final task in parallel.

While some rare tasks are by nature already parallel, most computational problems need to be parallelized, and this usually involves an intermediate preprocessing step and a subsequent information exchange between the nodes. One such special class of distributed computing algorithms follows the MapReduce model \cite{dean2008mapreduce}, which is a parallel processing tool that simplifies the parallel execution of tasks, by abstracting the original problem into the following three phases:
\begin{enumerate}
	\item the \textit{mapping phase}, where each element of the dataset is assigned to one or more computing nodes and where the nodes perform an intermediate computation aiming to ``prepare'' for parallelization,
	\item the \textit{shuffling phase} (or communication phase), where nodes communicate between each other the preprocessed data that is needed to make the process parallel, and
	\item the \textit{reduce phase}, where nodes work in parallel to provide the final output that each is responsible for.
\end{enumerate}
Classes of tasks that can be parallelized under a MapReduce framework include Sorting \cite{o2008terabyte}, Data Analysis and Clustering \cite{shim2012mapreduce,kumar2013verification}, Word Counting \cite{dean2007distributed}, Genome Sequencing \cite{mckenna2010genome}, and others.

\subsection{Communication bottleneck of distributed computing}
While though MapReduce allows for parallelization, it also comes with different bottlenecks involving for example struggling nodes~\cite{dean2010mapreduce} and non-fine-tuned algorithms~\cite{chen2011case}. The main bottleneck though that bounds the performance of MapReduce is the duration of the shuffling phase, especially as the dataset size becomes larger and larger. While having more nodes can speed up computational time, the aforementioned information exchange often yields unchanged or even increased communication load and delays, leading to a serious bottleneck in the performance of distributed computing algorithms.

\paragraph{Phase delays}
In particular, consider a setting where there are $K$ computing nodes, operating on a dataset of size $F$.
Assuming that each element of the dataset can appear in $t$ different computing nodes, and assuming that $T_{\text{map}}(F)$ represents the time required for one node to map the entire dataset, then the map phase will have duration approximately $T_{\text{map}}(t\frac{F}{K})$ which generally reduces with $K$.  Similarly the final reduce phase enjoys the same decreased delay $T_{\text{red}}(F/K)$, where $T_{\text{red}}(F)$ denotes the time required for a single node to reduce the entire mapped dataset\footnote{We here assume for simplicity of exposition, uniformity in the amount of mapped data that each node uses in the final reduce phase. We also assume a uniformity in the computational capabilities of each node.}.

The problem lies with the communication delay $T_{\text{com}}(F)$. For $T_c$ denoting the time required to transmit the entire mapped dataset, from one node to another without any interference from the other nodes\footnote{$T_c$ accounts for the ratio between the capacity of the communication link, and the dataset size $F$.}, and accounting for a reduction by the factor $(1-\gamma)$ due to the fact that each node already has a fraction $\gamma = t/K$ of the dataset, then the delay of the shuffling phase takes the form $T_{\text{com}}(F) = T_{c}\cdot(1-\gamma)$, which does not decrease with $K$.

Hence for the basic MapReduce (MR) algorithm --- under the traditional assumption that the three phases are performed sequentially --- the overall execution time becomes
\begin{equation*}
	T_{\text{tot}}^{\text{MR}}(F,K)=T_{\text{map}}\left(\frac{t}{K} F\right)+T_{c}\cdot(1-\gamma)+T_{\text{red}}\left(\frac{F}{K}\right)
\end{equation*}
which again shows that, while the joint computational cost $T_{\text{map}}(\frac{t}{K} F)+ T_{\text{red}}(\frac{F}{K})$ of the map and reduce phases can decrease by adding more nodes, the communication time $T_c\cdot(1-\gamma)$ is not reduced and thus the cost of the shuffling phase emerges as the actual bottleneck of the entire process.

\subsection{Emergence of Coded MapReduce: exploiting redundancy}

Recently a method of reducing the aforementioned communication load was introduced in \cite{7447112} (see also \cite{7965073,li2018fundamental}), which modified the mapping phase, in order to allow for the shuffling phase to employ coded communication. The main idea of the method --- which is referred to as Coded MapReduce (CMR) --- was to assign and then force each node to map a fraction $\gamma > 1/K$ of the whole dataset (such that each element of the dataset is mapped in $t = K\gamma$ computing nodes) and then --- based on the fact that such a mapping would allow for common mapped information at the different nodes --- to eventually perform coded communication where during the shuffling phase, the packets were not sent one after the other, but were rather combined together into XORs and sent as one. The reason this speedup would work is because the recipients of these packets could use part of their (redundant) mapped packets in order to remove the interfering packets from the received XOR, and acquire their own requested packet. This allowed for communicating (during the shuffling phase) to $t = K\gamma$ nodes at a time, thus reducing the shuffling phase duration, from $T_{c}\cdot(1-\gamma)$ to $\frac{1}{t}T_{c}\cdot(1-\gamma) = \frac{1}{K\gamma}T_{c}\cdot(1-\gamma) $.

\subsection{Subpacketization bottleneck of distributed computing} Despite the fact that the aforementioned coded method promises, in theory, big delay reductions by a factor of $t = K\gamma$ compared to conventional uncoded schemes, these gains are heavily compromised by the fact that the method requires that the dataset be split into an unduly large number of $S=t\binom{K}{t}$ packets which grows exponentially in $K$ and $t$.

Specifically the fact that the finite-sized dataset can only be divided into a finite number of packets, limits the values of parameter $t$ that can be achieved, because the corresponding subpacketization $S$ must be kept below some maximum allowable subpacketization $S_{\max}$, which, also, must be less than the number of elements $F$ in the dataset. If this number $S=t\binom{K}{t}$ exceeds the maximum allowable subpacketization $S_{\max}$, then coded communication is limited to include coding that spans only
\begin{equation}
\bar{K} = \arg\max_{K}\left\{t\binom{K}{t} \leq S_{\max}\right\}
\end{equation}
nodes at a time, forcing us to repeat coded communication $K/\bar{K}$ times, thus resulting in a smaller, actual gain $$\bar{t}=\bar{K}\gamma<K\gamma$$ which can be far below the theoretical communication gain from coding.
Such high subpacketization can naturally limit the coding gains $t$, but it can also further delay the shuffling phase because --- as we will elaborate later on --- it implies more transmissions and thus higher packet overheads, as well as because smaller packets are more prone to have mapped outputs that are unevenly sized, thus requiring more zero padding.

In what follows, we will solve the above problems with a novel group-based method of distributing the dataset across the computing nodes, and a novel method of cooperation/coordination between nodes in the transmission, which will jointly yield a much reduced subpacketization, allowing for a wider range of $t$ values to be feasible, thus eventually allowing substantial reductions in the overall execution time for a large class of distributed computing algorithms.

Before describing our solution and its performance, let us first elaborate on the exact channel model.

\subsection{Channel model: Distributed computing in a D2D setting}
In terms of the communication medium, we will focus on the wireless fully-connected setting, because in the wireless setting the nature of multicasting and the impact of link bottlenecks are clearer. As we will discuss later on though, the ideas here apply directly to the wired case as well.

We assume that the $K$ computing nodes are all fully connected via a wireless shared channel as in the classical fully-connected D2D wireless network. At each point there will be a set of active receivers, and active transmitters. Assuming a set of $L$ active transmitters jointly transmitting vector $\mathbf{x}\in\mathbb{C}^{L\times 1}$, then the received signal at a receiving node $k$ takes the form
\begin{equation}
	y_{k}=\mathbf{h}^{T}_ {k}\mathbf{x}+w_{k}, ~ ~ k=1,\cdots,K
\end{equation}
where as always $\mathbf{x}$ satisfies a power constraint $\mathbb{E}(||\mathbf{x}||^{2})<P$, where $\mathbf{h}_ {k}\in \mathbb{C}^{L\times 1}$ is the (potentially random) fading channel between the transmitting set of nodes and the receiving node $k$, and where $w_{k}$ denotes the unit-power AWGN noise at receiver $k$. We assume the system to operate in the high SNR regime (high $P$), and we assume perfect channel state information (CSI) (and for the wired case, perfect network coding coefficients) at the active receivers and transmitters.

\subsection{Notation}
We will use $[K]\triangleq \{1,2,\cdots,K\}$. If $\mathcal{A}$ is a set, then $|\mathcal{A}|$ will denote its cardinality, and $\mathcal{A}(j)$ will denote its $j$th element. For sets $\mathcal{A}$ and $\mathcal{B}$, then $\mathcal{A} \backslash \mathcal{B}$ denotes the difference set. For integers $n,k$ ($n\geq k$) then $\binom{n}{k}$ will denote the binomial ($n$-choose-$k$) operator. Complex vectors will be denoted by lower-case bold font.

\section{Main result}

We proceed to describe the performance of the new proposed algorithm, which will be presented in the next section. Key to this algorithm --- which we will refer to as the Group-based Coded MapReduce (GCMR) algorithm --- is the concept of user grouping. We will group the $K$ nodes into $K/L$ groups of $L$ nodes each, and then every node in a group will be assigned the same subset of the dataset and will produce the same mapped output. By properly doing so, this will allow us to use in the shuffling phase a new --- developed in this work --- D2D coded caching communication algorithm which assigns the D2D nodes an adaptive amount of content overlap\footnote{This general idea draws from the group-based cache-placement idea developed in \cite{LE18jsacSubmitted} for the cache-aided broadcast channel.}. This will in turn substantially reduce the required subpacketization, thus substantially boosting the speedup in communication and overall execution time.

For the sake of comparison, let us first recall that under the subpacketization constraint $S_{\max}$, the original Coded MapReduce approach achieves communication delay
\begin{equation}
T^{\text{CMR}}_{com}=\frac{1-\gamma}{\bar{t}}T_{c}
\end{equation}
where
\[ \bar{t} = \gamma \cdot \arg\max_{K}\{  K\gamma \binom{K}{K\gamma} \leq S_{\max}\}\] is the maximum achievable effective speedup (due to coding) in the shuffling phase. In the above and in what follows, we assume for simplicity that $Q=K$ such that each node has one final task.

We proceed with the main result.

\begin{theorem}
In the $K$-node distributed computing setting where the dataset can only be split into at most $S_{max}$ identically sized packets, the proposed Group-based Coded MapReduce algorithm with groups of $L$ users, can achieve communication delay \[T^{GCMR}_{com}=\frac{1-\gamma}{\bar{t}_{L}}T_{c}\] for \[\bar{t}_{L} = \gamma \cdot \arg\max_{K}\{\frac{K\gamma}{L}\binom{K/L}{K\gamma/L}\leq S_{\max}\}.\]
\end{theorem}
\begin{proof}
The proof follows directly from the description of the scheme in Section~\ref{sec:scheme}.
\end{proof}

The above implies the following corollary, which reveals that in the presence of subpacketization constraints, simple node grouping can further speedup the shuffling phase by a factor of up to $L$.
\begin{corollary}
In the subpacketization-constrained regime where $S_{\max}\leq \frac{K\gamma}{L}\binom{K/L}{K\gamma/L}$, the new algorithm here allows for shuffling delay
\[T^{\text{GCMR}}_{com}=\frac{1-\gamma}{\bar{t}_{L}}T_{c} = \frac{T^{\text{CMR}}_{com}}{L} \]
which is $L$ times smaller than the delay without grouping.

\end{corollary}
\begin{proof}
The proof is direct from the theorem.
\end{proof}

Finally the following also holds.
\begin{corollary}
When $S_{max}\geq \frac{K\gamma}{L}\binom{K/L}{K\gamma/L}$, the new algorithm allows for the unconstrained theoretical execution time\begin{equation}
T^{\text{GCMR}}_{tot}=T_{map}(\gamma F)+\frac{(1-\gamma)}{K\gamma}T_{c}+T_{red}\left(\frac{F}{K}\right).
\end{equation}
\end{corollary}
\begin{proof}
The proof is direct from the theorem.
\end{proof}

\section{Description of scheme}\label{sec:scheme}
We proceed to describe the scheme. We consider a dataset $\Phi$ consisting of $F$ elements, and a computational task that asks for $Q\geq K$ output values $u_{q}=\phi_{q}(\Phi), \ q=1,\cdots,Q$. 
The general aim is to distribute this task across the $K$ nodes, hence the dataset is split into $S$ disjoint packets $W_s, \ s=1,\cdots,S$ ($\cup_{s=1}^S W_s = \Phi$). We recall that, as is common in MapReduce, each function $\phi_{q}$ is decomposable as
\begin{equation}\label{model}
\phi_{q}(\Phi)=r_{q}(m_{q}(W_{1}),\cdots ,m_{q}(W_{S}))
\end{equation}
where the \emph{map functions} $\{m_{q}, \ q\in[Q]\}$ map packet $W_{s}$ into $Q$ intermediate values $W^{q}_{s}=m_{q}(W_{s}), \ q\in[Q]$, which are used by the reduce function $r_{q}$ to calculate the desired output value $u_{q}=r_{q}(W^{q}_{1},\cdots ,W^{q}_{S})$.

We proceed to describe the \textit{Assignment-and-Map, Shuffle} and \textit{Reduce} phases.

\subsection{Dataset assignment phase}
We split the $K$ nodes $1,2,\cdots,K$, into $K'\triangleq \frac{K}{L}$ groups
\begin{equation}
\mathcal{G}_{i}=\{i,i+K',...,i+(L-1)K'\},~~i\in[K']
\end{equation}
of $L$ nodes per group, and we split the dataset into
\begin{equation}\label{eq:subpacketization}
S=K'\gamma\binom{K'}{K'\gamma}
\end{equation}
packets, where $\gamma\in\{\frac{1}{K'},\frac{2}{K'},\cdots,1\}$ is a parameter of choice defining the redundancy factor of the mapping phase later on.
At this point, each $s = 1,2,\cdots,S$ is associated to a unique double index $\tau,\sigma$ so that the dataset can be seen as being segmented $\{W_{\tau,\sigma}, \ \tau\subseteq [K']:|\tau|=K'\gamma, \ \sigma\in\tau\}$.
Each node in group $\mathcal{G}_{i}$ is then assigned the set of packets
\begin{equation}
\mathcal{M}_{\mathcal{G}_{i}}=\{W_{\tau,\sigma}:\tau\ni i, \forall\sigma\in\tau\}
\end{equation}
and each of the $Q$ reduce functions $r_{q}$ is assigned to a given node. As noted before, for simplicity we assume that $Q=K$.

\subsection{Map Phase}
This phase consists of each node $k$ computing the map functions $m_{q}$ of all packets in $\mathcal{M}_{\mathcal{G}_{i}},\mathcal{G}_{i}\ni k$ for all $q\in[Q]$. At the end of the phase, node $k\in \mathcal{G}_{i}$ has computed the intermediate values $W^{q}_{\tau,\sigma}=m^{q}(W_{\tau,\sigma})$ for all $W_{\tau,\sigma}\in\mathcal{M}_{\mathcal{G}_{i}}$.

\subsection{Shuffle Phase}
Each node $\mathcal{G}_{i}(j)$ of group $\mathcal{G}_{i}$, must retrieve from the other nodes (except from those in $\mathcal{G}_{i}$), the intermediate values $\{W^{\mathcal{G}_{i}(j)}_{\tau,\sigma}: W_{\tau,\sigma}\notin \mathcal{M}_{\mathcal{G}_{i}}\}$ that it has not computed locally. Each node $\mathcal{G}_{i}(j)$ will thus create a set of symbols $\{x_{\mathcal{G}_{i}(j),\mathcal{Q}\setminus{\{i\}}}\}$, intended for all the nodes in groups $\mathcal{G}_{j}, j\in \mathcal{Q}\setminus{\{i\}}$ for some $\mathcal{Q} \subset [K']$ of size $|\mathcal{Q}|=K'\gamma+1$, where of course each symbol $x_{\mathcal{G}_{i}(j),\mathcal{Q}\setminus{\{i\}}}$ is a function of the intermediate values computed in the map phase. We use \[\mathbf{x}_{i,\mathcal{Q}\setminus{\{i\}}}\triangleq [x_{1,\mathcal{Q}\setminus{\{i\}}},\cdots ,x_{|\mathcal{M}_{\mathcal{G}_{i}}|,\mathcal{Q}\setminus{\{i\}}}]^{T}\] to denote the vector of symbols that are jointly created by the users in $\mathcal{G}_{i}$ and which are intended for the users in $\mathcal{G}_{j}, j\in \mathcal{Q}\setminus{\{i\}}$. Each symbol is communicated (multicasted) by the corresponding node $\mathcal{G}_{i}(j)$, to all the other nodes.
We proceed to provide the details for transmission and decoding.

\paragraph{Transmission}
For each subset $\mathcal{Q} \subset [K']$ of size $|\mathcal{Q}|=K'\gamma+1$, we sequentially pick all its elements $i\in\mathcal{Q}$ so that the users in group $\mathcal{G}_{i}$ act as a single distributed transmitter. These users in $\mathcal{G}_{i}$ construct the following vector of symbols
\begin{equation} \label{eq:transmission1}
\mathbf{x}_{i,\mathcal{Q}\setminus{\{i\}}}\!=\!\sum_{k' \in \mathcal{Q}\setminus{\{i\}}}\mathbf{H}^{-1}_{i,k'}
\begin{bmatrix}
W^{\mathcal{G}_{k'}(1)}_{Q\setminus{\{k'\}},i} ,
\cdots,
W^{\mathcal{G}_{k'}(L)}_{Q\setminus{\{k'\}},i}
\end{bmatrix}^{T}
\end{equation}
where $\mathbf{H}^{-1}_{i,k'}$ is the ZF precoding matrix for the channel $\mathbf{H}_{i,k'}\in\mathbb{C}^{L\times L}$ between transmitting group $\mathcal{G}_{i}$ and receiving group $\mathcal{G}_{k'}$, and where $\{W_{Q\setminus{\{k'\}},i}^{\mathcal{G}_{k'}(j)}\}_{j=1}^{L}$ is a set of intermediate values desired by the nodes in $\mathcal{G}_{k'}$. Each user $\mathcal{G}_{i}(j)$ now transmits the $j$-th element of the constructed vector $\mathbf{x}_{i,\mathcal{Q}\setminus{\{i\}}}$.

\paragraph{Decoding}
Node $\mathcal{G}_{p}(j), p\in \mathcal{Q}\setminus{\{i\}}$ receives the signal
\begin{equation}
y_{\mathcal{G}_{p}(j)}=\mathbf{h}^{T}_{\mathcal{G}_{p}(j)}\mathbf{x}_{i,\mathcal{Q}\setminus{\{i\}}}+w_{\mathcal{G}_{p}(j)}
\end{equation}
and removes out-of-group interference by employing the intermediate values it has computed locally in the map phase.
Specifically each node $\mathcal{G}_{p}(j)$, and all the nodes in $\mathcal{G}_{p}, p\in \mathcal{Q}$, remove from their $y_{\mathcal{G}_{p}(j)}$ the signal
\begin{equation}\mathbf{h}^{T}_{\mathcal{G}_{p}(j)}
\sum_{k' \in \mathcal{Q}\setminus{\{i,p\}}}\mathbf{H}^{-1}_{i,k'}
\begin{bmatrix}
W^{\mathcal{G}_{k'}(1)}_{Q\setminus{\{k'\}},i},  
\cdots,
W^{\mathcal{G}_{k'}(L)}_{Q\setminus{\{k'\}},i}
\end{bmatrix}^{T}
\end{equation}
to stay with a residual signal
\begin{equation}\label{eq:residual}
\mathbf{h}^{T}_{\mathcal{G}_{p}(j)}\mathbf{H}^{-1}_{\mathcal{G}_{i}, \mathcal{G}_{p}}\!
\begin{bmatrix}
W^{\mathcal{G}_{p}(1)}_{Q\setminus{\{p\}},i},
\cdots,
W^{\mathcal{G}_{p}(L)}_{Q\setminus{\{p\}},i}
\end{bmatrix}
^{T}
+w_{\mathcal{G}_{p}(j)}.
\end{equation}
By choosing $\mathbf{H}^{-1}_{\mathcal{G}_{i} , \mathcal{G}_{p}}$ to be a ZF precoder, removes intra-group interference, thus allowing each node $\mathcal{G}_{p}(j)$ to receive its desired intermediate value $W^{\mathcal{G}_{p}(j)}_{Q\setminus{\{p\}},i}$.  The shuffling phase is concluded by going over all the aforementioned sets $\mathcal{Q} \subset [K']$ of size $K'\gamma+1$.

\subsection{Reduce Phase}
At this point, each node uses the symbols received during the shuffling phase, together with the intermediate mapped values computed locally, in order to construct the inputs $W^{q}_{1},...,W^{q}_{S}$ that are required by the reduce function $r_{q}$ to calculate the desired output value $u_{q}=r_{q}(W^{q}_{1},\cdots ,W^{q}_{S})$.

\subsection{Calculation of Shuffling Delay}

We first see from \eqref{eq:subpacketization} that the subpacketization is, as stated, equal to
\begin{equation}\label{eq:subpacketization2}
S=K'\gamma\binom{K'}{K'\gamma} = \frac{K\gamma}{L}\binom{K/L}{K\gamma/L}.
\end{equation}

Let us now verify that the shuffling delay is $T^{GCMR}_{com} = \frac{1-\gamma}{\bar{t}_{L}}T_{c}$.
To do this, let us first assume that $S_{max} \geq S$ in which case we want to show that $T^{GCMR}_{com} = \frac{1-\gamma}{K\gamma}T_{c}$.
To verify the first term ($K\gamma$), we just need to note that during the shuffling phase no subfile is ever sent more than once, and then simply note that the scheme serves a total of $K'\gamma$ groups at a time, thus a total of $K'\gamma L = K\gamma$ nodes at a time. Finally to justify the term $1-\gamma$, we just need to recall that --- due to the placement redundancy --- a fraction $\gamma$ of all the shuffled data is already at their intended destination.

Lastly when $S_{max} \geq S$, we simply have to recall that we are allowed --- without violating the subpacketization constraint --- to encode over $\bar{K}_{L} = \arg\max_{K}\{\frac{K\gamma}{L}\binom{K/L}{K\gamma/L}\leq S_{max}\}$ nodes at a time, which yields the desired $\bar{t}_{L} = \gamma \cdot \bar{K}_{L}$.
This concludes the proof of the results.

\subsection{Extension to the Wired Setting}

As a last step, we quickly note that the same vector precoding used to separate the users of the same group (cf.~\eqref{eq:transmission1},\eqref{eq:residual}) can be directly applied in the wired setting where the intermediate nodes (routers, switches, etc.) in the links, can perform pseudo-random network coding operations on the received data (cf.~\cite{7580630}). This would then automatically yield a linear invertible relationship between the input vectors and the received signals, thus allowing for the design of the precoders that cancel intra-group interference.

 \begin{figure}[t!]
   \centering
   \includegraphics[width=0.28\textwidth]{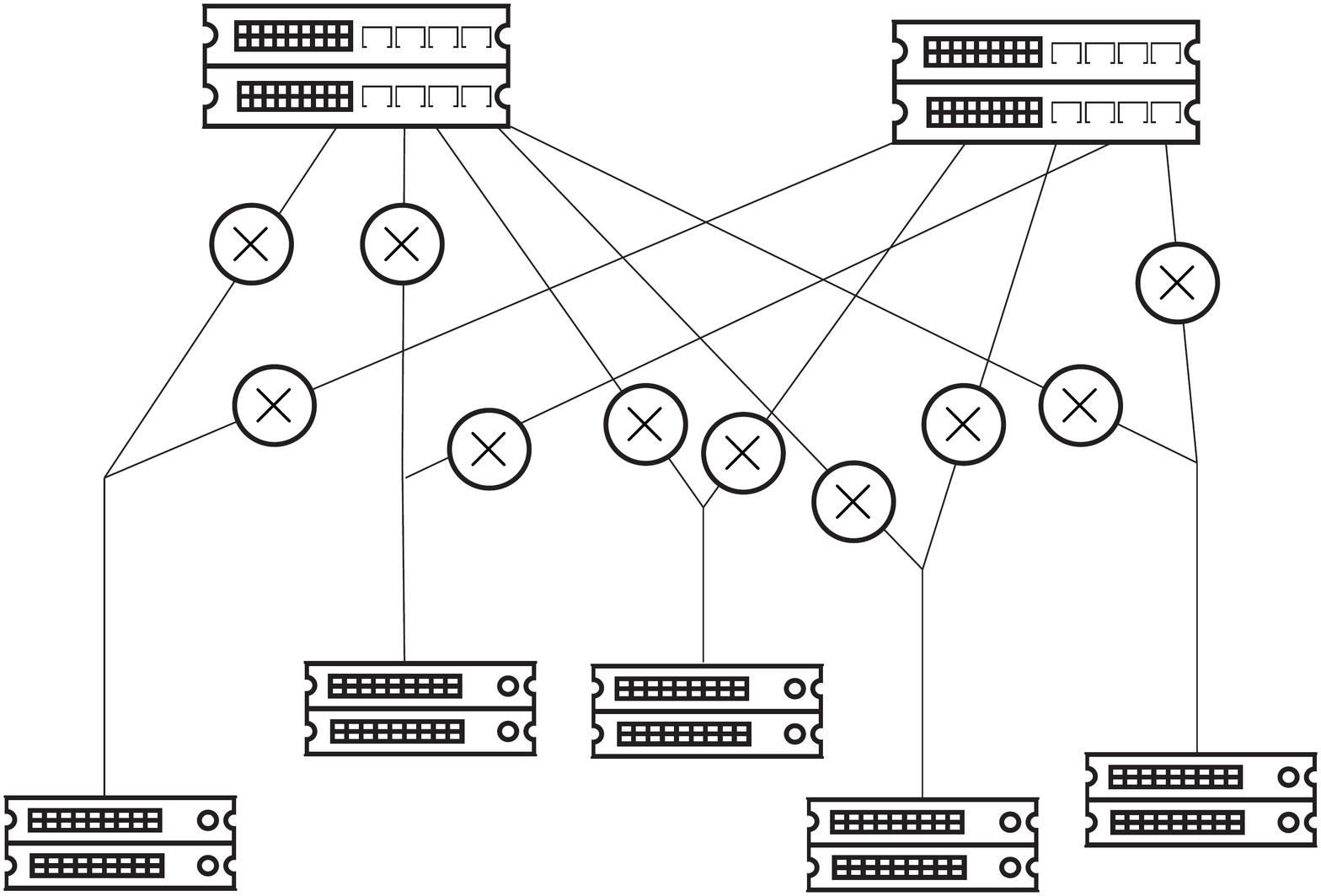}
   \caption{Illustration of the wired setting. $\times$ denotes a network coding operation.}
   \label{fig:wired}
 \end{figure}

\subsection{Example of the scheme}

Let us consider a setting with $K=32$ computing nodes, a chosen redundancy of $K\gamma=16$, and a cooperation parameter $L=8$. The nodes are split into $K/L=4$ groups
\begin{align*}
	\mathcal{G}_{1}=&\{1, 5, 9,  13, 17, 21, 25, 29\},\\[-2pt]
	\mathcal{G}_{2}=&\{2, 6, 10, 14, 18, 22, 26, 30\},\\[-2pt]
	\mathcal{G}_{3}=&\{3, 7, 11, 15, 19, 23, 27, 31\},\\[-2pt]
	\mathcal{G}_{4}=&\{4, 8, 12, 16, 20, 24, 28, 32\}
\end{align*}
and the dataset is split into $12$ packets as $\{W_{12,1},\allowbreak W_{12,2},\allowbreak W_{13,1},\allowbreak W_{13,3},\allowbreak W_{14,1},\allowbreak W_{14,4},\allowbreak W_{23,2},\allowbreak W_{23,3},\allowbreak W_{24,2},\allowbreak W_{24,4},\allowbreak W_{34,3},\allowbreak W_{34,4}\}$, which are distributed to the nodes of group $\mathcal{G}_{i}$ as follows:
\begin{align*}
	\mathcal{M}_{\mathcal{G}_{1}}=&\{W_{12,1}, W_{12,2}, W_{13,1}, W_{13,3}, W_{14,1}, W_{14,4}\}\\[-2pt]
	\mathcal{M}_{\mathcal{G}_{2}}=&\{W_{12,1}, W_{12,2}, W_{23,2}, W_{23,3}, W_{24,2}, W_{24,4}\}\\[-2pt]
	\mathcal{M}_{\mathcal{G}_{3}}=&\{W_{13,1}, W_{13,3}, W_{23,2}, W_{23,3}, W_{34,3}, W_{34,4}\}\\[-2pt]
	\mathcal{M}_{\mathcal{G}_{4}}=&\{W_{14,1}, W_{14,4}, W_{24,2}, W_{24,4}, W_{34,3}, W_{34,4}\}.
\end{align*}
In the map phase, each file $W_{\tau, \sigma}$ is mapped into $\{W^{q}_{\tau, \sigma}\}_{q=1}^{K}$ such that, for example, $W^{1}_{\tau, \sigma}$ is the output of the first mapping function after acting on $W_{\tau, \sigma}$. Finally the transmissions are\footnote{Please note that to keep the notation simple, the indices can often --- when there is not reason  for confusion --- appear without commas.}:
\begin{align*}
		\mathbf{x}_{1,23}=&\mathbf{H}_{12}^{-1}\mathbf{W}_{13,1}^{\mathcal{G}_{2}}+
	\mathbf{H}_{13}^{-1}\mathbf{W}_{12,1}^{\mathcal{G}_{3}}\\
		\mathbf{x}_{1,24}=&\mathbf{H}_{12}^{-1}\mathbf{W}_{14,1}^{\mathcal{G}_{2}}+
	\mathbf{H}_{14}^{-1}\mathbf{W}_{12,1}^{\mathcal{G}_{4}}\\
		\mathbf{x}_{1,34}=&\mathbf{H}_{13}^{-1}\mathbf{W}_{14,1}^{\mathcal{G}_{3}}+
	\mathbf{H}_{14}^{-1}\mathbf{W}_{13,1}^{\mathcal{G}_{4}}\\[3pt]
		\mathbf{x}_{2,13}=&\mathbf{H}_{21}^{-1}\mathbf{W}_{23,2}^{\mathcal{G}_{1}}+
	\mathbf{H}_{23}^{-1}\mathbf{W}_{12,2}^{\mathcal{G}_{3}}\\
		\mathbf{x}_{2,14}=&\mathbf{H}_{21}^{-1}\mathbf{W}_{24,2}^{\mathcal{G}_{1}}+
	\mathbf{H}_{24}^{-1}\mathbf{W}_{12,2}^{\mathcal{G}_{4}}\\
		\mathbf{x}_{2,34}=&\mathbf{H}_{23}^{-1}\mathbf{W}_{24,2}^{\mathcal{G}_{3}}+
	\mathbf{H}_{24}^{-1}\mathbf{W}_{23,2}^{\mathcal{G}_{4}}\\[3pt]
		\mathbf{x}_{3,12}=&\mathbf{H}_{31}^{-1}\mathbf{W}_{23,3}^{\mathcal{G}_{1}}+
	\mathbf{H}_{32}^{-1}\mathbf{W}_{13,3}^{\mathcal{G}_{2}}\\
		\mathbf{x}_{3,14}=&\mathbf{H}_{31}^{-1}\mathbf{W}_{34,3}^{\mathcal{G}_{1}}+
	\mathbf{H}_{34}^{-1}\mathbf{W}_{13,3}^{\mathcal{G}_{4}}\\
		\mathbf{x}_{3,24}=&\mathbf{H}_{32}^{-1}\mathbf{W}_{34,3}^{\mathcal{G}_{2}}+
	\mathbf{H}_{34}^{-1}\mathbf{W}_{23,3}^{\mathcal{G}_{4}}\\[3pt]
		\mathbf{x}_{4,12}=&\mathbf{H}_{41}^{-1}\mathbf{W}_{24,4}^{\mathcal{G}_{1}}+
	\mathbf{H}_{42}^{-1}\mathbf{W}_{14,4}^{\mathcal{G}_{2}}\\
		\mathbf{x}_{4,13}=&\mathbf{H}_{41}^{-1}\mathbf{W}_{34,4}^{\mathcal{G}_{1}}+
	\mathbf{H}_{43}^{-1}\mathbf{W}_{14,4}^{\mathcal{G}_{3}}\\
		\mathbf{x}_{4,23}=&\mathbf{H}_{42}^{-1}\mathbf{W}_{34,4}^{\mathcal{G}_{2}}+
	\mathbf{H}_{43}^{-1}\mathbf{W}_{24,4}^{\mathcal{G}_{3}},
\end{align*}
where $\mathbf{W}_{i,\tau}^{\mathcal{G}_{g}}$ denotes a vector of $L=8$ elements consisting of the intermediate values intended for nodes in group $\mathcal{G}_{g}$.

Observing for example the first transmission, we see that the nodes in group $\mathcal{G}_{2}$ can remove any interference caused by the intermediate values intended for group $\mathcal{G}_{3}$ since these intermediate values have been calculated by each node in $\mathcal{G}_{2}$ during the map phase. After noting that the precoding matrix $\mathbf{H}^{-1}_{12}$ removes intra-group interference, we can conclude that each transmission serves each of the $16$ users with one of their desired intermediate values, which in turn implies a $16$-fold speedup over the uncoded case.

\section{Conclusion}

The work provided a novel algorithm that employs node-grouping in the mapping and shuffling phases, to substantially reduce the shuffling-phase delays that had remained large due to the acute subpacketization bottleneck of distributed computing.

Among the most important contributions of this work is that, using node cooperation one, for the first time, can infinitely reduce the execution of these types of algorithms as long as there are enough computing nodes, something that previously wasn't possible in uncoded methods and while in coded methods it would reach a performance ceiling due to subpacketization constraints.

\subsection{Minimal overhead for group-based node cooperation}
It is interesting to note that the described node cooperation does not require any additional overhead communication of data (dataset entries) between the nodes. The only additional communication-overhead is that of having to exchange CSI between active receiving and transmitting nodes from $K\gamma/L+1$ groups. In static settings --- where computing nodes are not moving fast, as one might expect to happen in data centers --- and in particular in wired settings where the network coding coefficients are fixed and known, the CSI overhead can be very small compared to the volumes of the communicated datasets.

\subsection{Impact of reducing subpacketization on distributed computing}
We have have seen how extremely large subpacketization requirements can diminish the effect of coding in reducing the shuffling-phase delays. The proposed algorithm allows --- with minimal or no additional overhead --- for a dramatically reduced subpacketization, which comes with several positive ramifications.\setcounter{paragraph}{0} \paragraph{Boosting the Speedup-Factor $t$ in the Shuffling Phase} As we have discussed, the much reduced subpacketization allows for a substantial increase in the number of nodes we can encode over, thus potentially yielding an $L$-fold decrease in the shuffling-phase delay. The fact that a finite-sized dataset can only be divided into a finite number of subpackets, limits the values of parameter $t$ that can be achieved, because the corresponding subpacketization, which need be as high as $S=t\binom{K}{t}$, must be kept below some maximum subpacketization $S_{\max}$, which itself must be substantially less than the total number of elements $F$ in the dataset. When this number $S=t\binom{K}{t}$ exceeds the maximum allowable subpacketization $S_{\max}$, then what is often done is that coded communication is limited to include coding that spans only $\bar{K}$ users at a time (thus coded communication is repeated $K/\bar{K}$ times, for some $\bar{K}$ that satisfies $\bar{K}\gamma\binom{\bar{K}}{\bar{K}\gamma}\le S_{\max})$, thus resulting in a smaller, actual, gain $\bar{t}=\bar{K}\gamma<K\gamma$, which can be far below the theoretical communication gain from coding.

\paragraph{Reducing Packet Overheads} The second ramification from having fewer packets, comes in the form of reduced header overheads that accompany each transmission.
As the subpackets --- and thus their combinations --- become smaller and smaller, which means that the overhead ``headers'' that must accompany each transmission, will occupy a significant portion of the transmitted signal. Simply put, the more the subpackets, the smaller they are, hence the more the communication load is dominated by header overheads.
	\paragraph{Efficient Coded Message Creation by Reducing Unevenness} Another positive ramification from our algorithm is that it can reduce the unevenness between the sizes of the mapped outputs that each packet is mapped into. This unevenness --- which is naturally much more accentuated in smaller packets --- can cause substantial additional delays because it forces zero padding (we can only coombine equal-sized bit streams) which wastes communication resources. Having fewer and thus larger packets, averages out these size variations, thus reducing wasteful zero padding.

This can be better understood by using the Terasort and Coded Terasort framewoks \cite{o2008terabyte,7965073} for sorting $F$ numbers, by making use of $K=3$ nodes and having a chosen redundancy of $t=K\gamma=2$, but instead of assuming that each intermediate value has equal amount of elements, i.e., instead of assuming that $|W_i^1 |=|W_i^2 |=|W_i^3 |=1/3 |W_i |=F/18  ,i=1,2,…,6$, (recall that each of the $6$ subpackets has size $|W_i |=F/6$) we will instead assume that any intermediate value $W_1^3,W_2^3,W_3^3,W_4^3,W_5^3,W_6^3$ with upper index 3, will each occupy a fraction 1/2 of the elements of the respective subpacket (i.e., $|W_i^3 |=1/2 |W_i |=W/12  ,i=1,2,…,6$), while intermediate values with upper index $1$ or $2$ $(W_1^1,W_2^1,W_3^1,W_4^1,W_5^1,W_6^1)$ and $(W_1^2,W_2^2,W_3^2,W_4^2,W_5^2,W_6^2)$, will only have 1/4 of the elements of their respective subpacket each (i.e., $|W_i^1 |=1/4 |W_i |=F/24,i=1,2,…,6$), and $|W_i^2 |=1/4 |W_i |=F/24,~i=1,2,…,6$. In the case of uncoded placement, the corresponding delay would remain $(1-\gamma) T_c=(1-2/3) T_c=  1/3 T_c$ because there are no XORs, and because despite the unevenness, the total amount of information that must be communicated, remains the same. On the other hand, in the case of coded communication, having $|W_i^1 |=|W_2^1 |=1/4 |W_i |=F/12\neq |W_i^3 |=1/2 |W_i |=F/6$, in turn means that for every aforementioned XOR $x_1=W_1^2\oplus W_3^3, x_2=W_4^3\oplus W_5^1$ that includes some of the $W_{i}^{3}, i\in\{1,2,3\}$ elements inside, we would have to perform zero padding; for example, in the case of $x_2=W_4^3\oplus W_5^1$, we would have to zero pad $W_5^1$ to double its size, thus wasting resources. Now the three introduced XORs ($x_1=W_1^2\oplus W_3^3, x_2=W_4^3\oplus W_5^1, x_3=W_2^2\oplus W_6^1$) will have sizes $|x_1 |=|x_2 |=F/12, |x_3 |=F/24$, and thus sending all three would require a total delay of $T_c/12+T_c/12+T_c/24= 5T_c/24$. 

Comparing the above to the delay $1/3 T_c$ of the uncoded case, we can see that the multiplicative gain in the communication phase – due to coded communication \cite{7447112} – is limited to Gain = $(1/3)/(5/24)= 8/5=1.6$, instead of the theoretical gain of $t=2$.
On the other hand, by decreasing subpacketization, we automatically increase the size of these subpackets, thus decreasing – with high probability, due to the law of large numbers – the relative unevenness, which in turn allows for higher speedup gains.

\bibliographystyle{ieeetr}

\end{document}